\documentclass[letterpaper, 10 pt, conference]{ieeeconf}  

\IEEEoverridecommandlockouts                              
\usepackage{common_pkgs}

\title{\LARGE \bf
Minimal L2-Consistent Data-Transmission
}

\author{Antoine Aspeel, Laurent Bako, and Necmiye Ozay
\thanks{This work is funded by the ONR grant N00014-21-1-2431 (CLEVR-AI).}
\thanks{A. A. and N. O. are with the Electrical Engineering and Computer Science Department, Univ. of Michigan, Ann Arbor, MI {\tt\small \{antoinas,necmiye\}@umich.edu}}%
\thanks{L. B. is with Ecole Centrale de Lyon, INSA Lyon, Université Claude Bernard Lyon 1, CNRS, Ampère, UMR 5005, 69130 Ecully, France.
        {\tt\small laurent.bako@ec-lyon.fr}}%
}

\begin{document}

\maketitle
\thispagestyle{empty}
\pagestyle{empty}

\begin{abstract} In this work, we consider non-collocated sensors and actuators, and we address the problem of minimizing the number of sensor-to-actuator transmissions while ensuring that the L2 gain of the system remains under a threshold. By using causal factorization and system level synthesis, we reformulate this problem as a rank minimization problem over a convex set. When heuristics like nuclear norm minimization are used for rank minimization, the resulting matrix is only numerically low rank and must be truncated, which can lead to an infeasible solution. To address this issue, we introduce approximate causal factorization to control the factorization error and provide a bound on the degradation of the L2 gain in terms of the factorization error. The effectiveness of our method is demonstrated using a benchmark.
\end{abstract}

\section{INTRODUCTION}

Modern real-world systems are often made up of several components interconnected via a communication network. This work considers networked control systems where sensors and actuators are physically distant and need to communicate through a network to implement a control law. Examples of relevant applications include smart building heating systems, disaster relief operation systems, and drone control in a motion capture arena. In all these scenarios, the sensors information must be communicated to the actuators in order to implement a feedback law. It is therefore desired to optimize the number of transmissions while maintaining an acceptable level of control performance.

More precisely, this work focuses on the case of sensors and actuators that are not collocated, and seek to minimize the number of transmissions from sensors to actuators, while keeping the L2 gain between the external disturbances and the system state under a given threshold.

\subsubsection{Related works}

There has been a growing number of works addressing a variety of challenges related to computation, sampling and communication efficiency for resource-constrained systems. In the last decade, the event-triggered approach has been proposed as a way to minimize the use of actuators \cite{heemels2012introduction}. These methods incorporate a triggering mechanism which determines appropriate actuation or transmission times by monitoring the state of the system. The impact of such triggering mechanisms on the L2 gain is studied in \cite{strijbosch2020mathcal, luo2022stability, wu2021l2}.

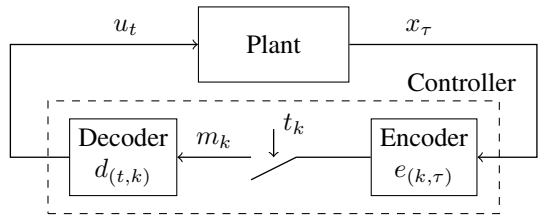
\begin{figure}[ht]
\centering
\begin{tikzpicture}[scale=1]
\centering
\def\boxHight{1cm}
\def\plantWidth{2cm}
\def\xlim{3.5cm}
\def\encCenter{2cm}
\def\encWidth{1cm}
\def\ylim{1.5cm}
\def\KHight{1.5cm}
\def\KWidth{6.0cm}
\node [rectangle, draw, minimum width=\plantWidth, minimum height=\boxHight, align=center] (plant) at (0,0) {Plant}; 
\node [rectangle, draw, minimum width=\encWidth, minimum height=\boxHight, align=center] (encoder) at (\encCenter,-\ylim) {Encoder\\ $e_{(k,\tau)}$}; 
\node [rectangle, draw, minimum width=\encWidth, minimum height=\boxHight, align=center] (decoder) at (-\encCenter,-\ylim) {Decoder\\ $d_{(t,k)}$}; 
\node [dashed, rectangle, draw, minimum width=\KWidth, minimum height=\KHight] (K) at (0,-\ylim) {}; 

\draw [line width=0.5] (encoder) -- (0.3cm,-\ylim) -- (-0.3cm,-\ylim-0.3cm); 
\draw [->,line width=0.5] (-0.3cm,-\ylim) -- node[above]{$m_k$} (decoder); 
\draw [->,line width=0.5] (0,-\ylim+\ylim/4) -- node[above right]{$t_k$} (0,-\ylim); 
\draw [->,,line width=0.5] (plant) -- node[above left]{$x_\tau$} (\xlim,0) -- (\xlim,-\ylim) -- (encoder); 
\draw [->,,line width=0.5] (decoder) -- (-\xlim,-\ylim) -- (-\xlim,0) -- node[above right]{$u_t$} (plant); 

\draw (\encCenter+0.5cm,-\ylim+\KHight/2+0.25cm) node{Controller};
\end{tikzpicture}
\vspace{-0cm}
\caption{Encoder-decoder structure of the controller.}
\vspace{-.1cm}
\label{fig:block_diagram}
\end{figure}

The concept of L2-consistent data transmission is introduced and studied in \cite{balaghi20192, balaghiinaloo2021l2, antunes2019consistent}. A sequence of transmission times is said to be L2-consistent if it performs at least as well as a given reference periodic transmission scenario in terms of both L2 gain and number of transmissions. The finite horizon setting is considered in \cite{balaghi20192} and an event-triggered extension is proposed in \cite{balaghiinaloo2021l2, antunes2019consistent} for infinite horizons. In both cases, a method based on dynamic games is proposed to find L2-consistent transmission times.

All these works consider that the controller is either (i) on the sensor side, i.e., the control input is computed on the sensor side and sent to the actuator when the triggering condition is met; or (ii) on the actuator side, i.e., the state is sent to the actuator  where the control input is computed. In contrast, a key distinguishing feature of our approach is that the controller is split in two parts (an encoder and a decoder), allowing for computations on both the sensor and the actuator sides. We show that this controller structure allows to reduce the number of transmissions even further.

The idea of encoder-decoder structure that splits a controller into a sensor-side and an actuator-side controller appears in \cite{braksmayer2017redesign, braksmayer2019discrete}, where a method to approximate a linear controller to accommodate given sensor-to-actuator transmission times is proposed. This structure is also considered in our previous work \cite{aspeel2023low} in the context of safety control, where we show that the number of sensor-to-actuator transmissions is given by the rank of the matrix mapping the sequence of measurements to the sequence of control inputs. Consequently, minimizing the  number of transmissions reduces to a rank minimization problem. Finally, it is shown in \cite{aspeel2023low} that the encoder and decoder are given by the causal factorization of that matrix. Also somewhat related, in \cite{cho2023low}, several distributed agents implement a low-rank time-invariant feedback gain to minimize a LQR cost through limited broadcast communication. The goal in that paper is not to minimize the number of transmission times but rather to send compact messages at each time.

\subsubsection{Contributions}

We consider a discrete-time linear system over a finite horizon, and we address the problem of finding a linear controller with memory that can be implemented with a minimum number of sensor-to-actuator transmissions, while keeping the L2 gain below a given threshold.

Our first contribution is to prove that an optimal transmission schedule and a controller can be co-designed by solving a rank minimization problem over a convex set, where the rank gives the number of transmissions. Then, we use causal factorization introduced in \cite{aspeel2023low} to obtain the controller implementation.  

When heuristics are used for rank minimization, they typically lead to solutions that are full rank but have several small singular values. If used ``as is”, such solutions would lead to a number of transmissions equal to their exact rank which is likely to be large. To solve this issue, our second contribution is to introduce the notion of \emph{approximate causal factorization} and an algorithm to compute it. Finally, our third contribution is to bound the deterioration of the L2 gain as a function of the factorization error. Finally, we show that our method results in fewer transmissions compared to a benchmark method.

\subsubsection{Notation}
The identity matrix in $\R^{n\times n}$ is denoted $I_n$, or $I$ when the dimension is clear from the context. For a matrix $\mathbf{X}$, $\|\X\|$ is its induced 2-norm. The Kronecker product between two matrices $\mathbf{X}$ and $\mathbf{Y}$ is denoted $\mathbf{X}\otimes\mathbf{Y}$. For a matrix $\X$, $\X_{i,:}$ denotes the $i$-th row of $\X$, and $\X_{i:j,:}$ denotes the submatrix formed by its rows indexed from $i$ to $j$. For columns, the notations $\X_{:,i}$ and $\X_{:,i:j}$ are defined similarly. A matrix $\X\in\R^{Tm\times Tn}$ is $(m,n)$-block-lower-triangular if $\X_{tm+1:(t+1)m,\tau n+1:(\tau+1)n}=0$ for all $0\leq t<\tau \leq T-1$.

\section{PROBLEM STATEMENT}

Consider the system dynamics
\begin{align}\label{eq:dynamics}
x_{t+1}=Ax_t+Bu_t+Dw_t
\end{align}
over a finite horizon $T$ with state $x_t\in\R^{n_x}$, control input $u_t\in\R^{n_u}$, process noise $w_t\in\R^{n_w}$, and unknown initial state $x_0$.

To reduce the number of sensor-to-actuator transmissions, we want to design a controller that can be implemented via an encoder-decoder structure as proposed in \cite{aspeel2023low} (see Fig.~\ref{fig:block_diagram}).

On the sensor side, a linear encoder with a memory computes messages $m_k$ based on the previous states $x_\tau$. The message $m_k$ is transmitted at time $t_k$ to the decoder located on the actuator side. The linear decoder computes control inputs $u_t$ based on the messages previously received. For $r$ transmission times
$\{t_k\}_{k=1}^r$ such that $0\leq t_1 \leq t_2 \leq \dots \leq t_r \leq T$, the controller structure is written
\begin{align}\label{eq:controller}
m_k = \sum_{\tau\leq t_k} e^\top_{(k,\tau)}x_\tau\text{ and }
u_t = \sum_{k \text{ s.t. }t_k\leq t} d_{(t,k)}m_k,
\end{align}
where each message $m_k\in\R$ is a real number. This message encodes information from states through the vectors $e_{(k,\tau)}\in\R^{n_x}$. Each input $u_t$ is computed by decoding the messages using the vectors $d_{(t,k)}\in\R^{n_u}$. Importantly, this controller is causal since every message is transmitted (i) after the states it encodes have been measured, but (ii) before being used to compute an input. The number of transmissions required to implement this controller is the number $r$ of transmitted messages.

We are interested in finding a controller that leads to an L2 gain smaller than a given $\gamma\geq0$, that is
\begin{align}\label{eq:l2:def}
\sum_{t=0}^{T} x_t^\top Q x_t + u_t^\top R u_t \leq \gamma^2 \sum_{t=0}^{T}w_t^\top w_t, \text{ for all } x_0,\  w_t,
\end{align}
with the matrices $Q$ and $R$ being symmetric positive semidefinite. Subject to this gain constraint, we seek to find a controller that minimizes the number $r$ of transmissions between the encoder/decoder parts.

\begin{problem}\label{prob}
Find the minimal $r$ such that there exist $\{t_k\}_{k=1}^r$, $\{e_{(k,\tau)}\}_{k=1,\dots,r}^{\tau=0,\dots,t_k}$ and $\{d_{(t,k)}\}_{k=1,\dots,r}^{t=t_k,\dots,T}$ satisfying \eqref{eq:dynamics}, \eqref{eq:controller} and \eqref{eq:l2:def}.
\end{problem}

\begin{remark}
In \cite{balaghi20192}, a sequence of transmission times is said to be \emph{L2-consistent for a period $p$} if compared to the $p$-periodic transmission, it leads to (i) a not larger L2 gain, and (ii) no more transmissions. Note that no encoder is considered in that work: when a transmission occurs, the state $x_t$ is transmitted, which corresponds to $n_x$ messages. If $\gamma_p$ is the minimum L2 gain achievable with $p$-periodic transmissions, then solving Problem~\ref{prob} with $\gamma=\gamma_p$ in \eqref{eq:l2:def} leads to L2-consistent transmission times for the period $p$.
\end{remark}

\section{EQUIVALENCE WITH RANK MINIMIZATION}

In this section, we show that using techniques from~\cite{aspeel2023low}, Problem~\ref{prob} can be reduced to a rank minimization over a convex set. We note that while~\cite{aspeel2023low} considers safety over bounded disturbances, we are interested in the case of a bound on the L2 gain. In Subsection~\ref{sec:causal_factorization}, Problem~\ref{prob} is reduced to a rank minimization problem over a non-convex set; then, in Subsection~\ref{sec:SLS}, the system level synthesis framework~\cite{anderson2019system} is used to make the constraints convex.

\subsection{Causal factorization}\label{sec:causal_factorization}
Let us introduce the following notation: $\mathbf{u}\coloneqq\begin{bmatrix} u_0^\top & \dots & u_T^\top\end{bmatrix}^\top$, $\mathbf{x}\coloneqq\begin{bmatrix} x_0^\top & \dots & x_T^\top\end{bmatrix}^\top$, $\mathbf{m}\coloneqq\begin{bmatrix} m_1 & \dots & m_r \end{bmatrix}^\top$ and
\begin{align*}
    &\mathbf{D} \coloneqq\begin{bmatrix}
    d_{(0,1)} & \cdots & d_{(0,r)} \\
    \vdots & & \vdots \\
    d_{(T,1)} & \cdots & d_{(T,r)}
    \end{bmatrix},\ 
    \mathbf{E} \coloneqq\begin{bmatrix}
    e^\top_{(1,0)} & \cdots & e^\top_{(1,T)} \\
    \vdots & & \vdots \\
    e^\top_{(r,0)} & \cdots & e^\top_{(r,T)}
    \end{bmatrix},
\end{align*}
with $d_{(t,k)}\coloneqq0$ when $t<t_k$, and $e_{(k,\tau)}\coloneqq0$ when $t_k<\tau$. With these notations, the controller structure \eqref{eq:controller} can be written as
\begin{align}\label{eq:controller:DE}
\mathbf{m=Ex}\text{ and }\mathbf{u=Dm}.
\end{align}
One can optimize over the matrices $\D$ and $\E$ instead of $e_{(k,\tau)}$ and $d_{(t,k)}$ as long as the pair $(\D,\E)$ satisfies the following constraint:
\begin{constraint}[Causality]\label{constraint:causality}
There exist $0\leq t_1 \leq t_2 \leq \dots \leq t_r \leq T$ such that $\mathbf{E}_{k,\tau n_x+j} = \mathbf{D}_{tn_u+i,k} = 0$ for all $k=1,\dots,r$, for all $t<t_k<\tau$, for all $i=1,\dots,n_u$, and for all $j=1,\dots,n_x$.
\end{constraint}

This constraint ensures the causality of the controller: the message $m_k$ transmitted at time $t_k$ does not encode states $x_\tau$ received after $t_k$, and the control input $u_t$ does not depend on messages transmitted after $t$. This constraint is generally hard to enforce because of the ``there exist" quantification on the transmission times $t_k$.

Abstracting the encoder and the decoder, the controller structure \eqref{eq:controller:DE} can be written as
\begin{align} \label{eq:controller:K}
u_t=\sum_{\tau\leq t}K_{(t,\tau)}x_\tau,
\end{align}
with $K_{(t,\tau)}=\sum_{k\text{ s.t. }\tau\leq t_k\leq t} d_{(t,k)}e^\top_{(k,\tau)}$. Writing
\begin{align} \label{eq:K:block}
    &\K\coloneqq\begin{bmatrix}
    K_{(0,0)} & & & \\
    K_{(1,0)} & K_{(1,1)} & & \\
    \vdots & \ddots & \ddots & \\
    K_{(T,0)} & \hdots & K_{(T,T-1)} & K_{(T,T)}
    \end{bmatrix},
\end{align}
leads to the controller $\mathbf{u=Kx}$ where $\K=\D\E$ is $(n_u,n_x)$-block-lower triangular. The matrices $\D$ and $\E$ can be recovered from $\K$ via its causal factorization.

\begin{definition}[Causal factorization \cite{aspeel2023low}]\label{def:causal_facto}
Let $\X\in\R^{(T+1)n_u\times (T+1)n_x}$ be a $(n_u,n_x)$-block-lower triangular matrix. A pair of matrices $(\D,\E)\in\R^{(T+1)n_u\times r}\times \R^{r \times (T+1)n_x}$ is a \emph{causal factorization of $\X$ with band $r$} if $\X = \D\E$ and Constraint \ref{constraint:causality} (Causality) holds.
\end{definition}

Because of the condition $\mathbf{X=DE}$, a causal factorization of $\X$ can not have a band smaller than $\rank~\X$. In \cite{aspeel2023low}, it is proven that there is always a causal factorization with band exactly $\rank~\X$ and an algorithm is provided to compute it. As a consequence, the problem of minimizing the number of messages reduces to a rank minimization problem followed by the computation of a causal factorization.

\begin{lemma}\label{thm:prob:K}
Optimal $r$, $\{t_k\}_{k=1}^r$, $\{e_{(k,\tau)}\}_{k=1,\dots,r}^{\tau=0,\dots,t_k}$ and $\{d_{(t,k)}\}_{k=1,\dots,r}^{t=t_k,\dots,T}$ for Problem~\ref{prob} are obtained by finding an optimal $\K^*$ for
\begin{align}\label{eq:prob:K}
\min_{\K} \rank \K \text{ s.t. \eqref{eq:dynamics}, \eqref{eq:l2:def}, \eqref{eq:controller:K}, \eqref{eq:K:block}},
\end{align}
and computing a causal factorization of $\K^*$ with band equal to $\rank \K^*$.
\end{lemma}
\begin{proof}
The proof is analogous to the proof of \cite[Corollary~1]{aspeel2023low}. Equations \eqref{eq:dynamics}, \eqref{eq:l2:def}, \eqref{eq:controller:K}, \eqref{eq:K:block} correspond to (1), Constraint~1, (4) and (5) in~\cite{aspeel2023low}, respectively.
\end{proof}

\subsection{System level synthesis}\label{sec:SLS}

We note that the L2 constraint \eqref{eq:l2:def} in problem~\eqref{eq:prob:K} is not convex in $\K$. In this subsection, we show that system level synthesis (SLS) can be used to rewrite the optimization problem~\eqref{eq:prob:K} as a rank minimization over a convex set.
For this purpose, let us introduce the vector $\mathbf{w}=\begin{bmatrix}x_0^\top&w_0^\top & \dots & w_{T-1}^\top \end{bmatrix}^\top$ and the matrices $\mathcal{A}\coloneqq I_{T+1}\otimes A$, $\mathcal{B}\coloneqq I_{T+1}\otimes B$, and $\mathcal{D}\coloneqq \text{blkdiag}(I_{n_x}, I_{T}\otimes D)$, where $\text{blkdiag}$ indicates a block-diagonal concatenation. Finally, let $Z\in\R^{(T+1)n_x\times(T+1)n_x}$ be the block-downshift operator, i.e., the $(n_x,n_x)$-block-lower triangular matrix with identities on its first block subdiagonal and zeros elsewhere. This allows us to write
\begin{align}\label{eq:system_response}
\begin{bmatrix} \mathbf{x} \\ \mathbf{u} \end{bmatrix} = \begin{bmatrix}\Phi_x \\ \Phi_u\end{bmatrix}\mathcal{D}\mathbf{w},
\end{align}
where $\Phi_x=(I-Z(\mathcal{A}+\mathcal{B}\K))^{-1}$ and $\Phi_u=\K(I-Z(\mathcal{A}+\mathcal{B}\K))^{-1}$ are the system responses. \\
The following result forms the core of SLS.
\begin{lemma}[{\cite[Theorem 2.1]{anderson2019system}}]\label{thm:SLS}
Over the horizon $t=0,\dots,T$, the system dynamics \eqref{eq:dynamics} with $(n_u,n_x)$-block-lower triangular state feedback law $\K$ defining the control action as $\mathbf{u=Kx}$, the following are true:
\begin{enumerate}
\item the affine subspace defined by
\begin{align}\label{eq:SLC}
\begin{bmatrix}I-Z\mathcal{A} & -Z\mathcal{B}\end{bmatrix}\begin{bmatrix}\Phi_x \\ \Phi_u\end{bmatrix} = I,
\end{align}
parameterizes all possible system responses \eqref{eq:system_response}.

\item for any $(n_x,n_x)$- and $(n_u,n_x)$-block-lower triangular matrices $\{\Phi_x,\Phi_u\}$ satisfying \eqref{eq:SLC}, the controller $\K=\Phi_u\Phi_x^{-1}$ achieves the desired system response \eqref{eq:system_response}. 
\end{enumerate}
\end{lemma}

Introducing the matrices $\mathcal{Q}\coloneqq I_T\otimes Q$ and $\mathcal{R}\coloneqq I_T\otimes R$, SLS allows to rewrite the constraint \eqref{eq:l2:def} on the L2 gain as
\begin{align}\label{eq:l2:sls}
\left\| \begin{bmatrix} \mathcal{Q}^{1/2} & 0 \\ 0 & \mathcal{R}^{1/2} \end{bmatrix} \begin{bmatrix}\Phi_x \\ \Phi_u\end{bmatrix}\mathcal{D} \right\|\leq \gamma,
\end{align}
which is convex in $\{\Phi_x,\Phi_u\}$.

The following result shows that minimizing the rank of $\K$ can be achieved using SLS by minimizing the rank of $\Phi_u$. This is our first main contribution.

\begin{theorem}\label{thm:min_rank_SLS}
Optimal $r$, $\{t_k\}_{k=1}^r$, $\{e_{(k,\tau)}\}_{k=1,\dots,r}^{\tau=0,\dots,t_k}$ and $\{d_{(t,k)}\}_{k=1,\dots,r}^{t=t_k,\dots,T}$ for Problem~\ref{prob} are obtained by finding optimal $\{\Phi_x,\Phi_u\}$ for
\begin{align}\label{eq:prob:SLS}
\min_{\Phi_x,\Phi_u} \rank \Phi_u &\text{ s.t. \eqref{eq:SLC}, \eqref{eq:l2:sls}, and}\notag\\
&\Phi_x,\ \Phi_u \text{ are block-lower triang.},
\end{align}
and computing a causal factorization of the matrix $\K\coloneq\Phi_u\Phi_x^{-1}$.
\end{theorem}
\begin{proof}
From Lemma~\ref{thm:prob:K}, Problem~\ref{prob} can be solved by solving \eqref{eq:prob:K} for $\K$. Then, from Lemma~\ref{thm:SLS}, solving \eqref{eq:prob:K} is equivalent to solve $ \min_{\Phi_x,\Phi_u} \rank \left(\Phi_u\Phi_x^{-1}\right) \text{ s.t. \eqref{eq:SLC}, \eqref{eq:l2:sls}}$, for block-lower triangular $\{\Phi_x,\Phi_u\}$. Note that any block-lower triangular matrix $\Phi_x$ that satisfies \eqref{eq:SLC}, has identity matrices on its block diagonal and is therefore invertible. Finally, $\rank (\Phi_u\Phi_x^{-1})=\rank \Phi_u$, which concludes the proof.
\end{proof}
Overall, Theorem~\ref{thm:min_rank_SLS} reduces Problem~\ref{prob} to a rank minimization over a convex set, followed by the computation of a causal factorization.

\section{FROM FACTORIZATION TO TRUNCATION}

While Theorem~\ref{thm:min_rank_SLS} theoretically gives a method to solve Problem~\ref{prob}, it presents some numerical challenges. Indeed, minimizing the rank is NP-hard and one typically relies on heuristics such as nuclear norm relaxation (which is the tightest convex relaxation of the rank \cite[Theorem~1]{fazel2001rank}). This leads to a numerical solution $\{\Phi_x^{\text{num}},\Phi_u^{\text{num}}\}$ for which $\Phi_u^{\text{num}}$ is only \emph{numerically} low rank, i.e., it has many small (but non-zero) singular values. These matrices lead to a controller $\K^{\text{num}}\coloneqq \Phi_u^{\text{num}}(\Phi_x^{\text{num}})^{-1}$ that is -- at best -- numerically low rank as well.

To obtain a controller that requires a number of transmissions equal to the numerical rank, and not the exact rank, some approximation is needed. A usual way to find a low-rank approximation of a matrix is via truncated singular value decomposition (SVD). However a truncated SVD of $\K^{\text{num}}$ would lead to a matrix that is not block-lower triangular and for which no causal factorization exists.

This section aims to solve this issue. In Subsection~\ref{sec:approx_causal_facto}, we introduce the notion of \emph{approximate causal factorization} for which $\K^{\text{num}}\approx\D\E$. Allowing for a factorization error makes it possible to reduce the band below the (exact) rank of $\K^{\text{num}}$. Then, in Subsection~\ref{sec:robustness} we derive an upper bound on the degradation of the L2 gain in terms of the factorization error.

\subsection{Approximate causal factorization}\label{sec:approx_causal_facto}

In this subsection, we define formally the notion of approximate causal factorization (or $\epsilon$-causal factorization) and we provide an algorithm to compute it.

\begin{definition}[$\epsilon$-causal factorization]
Let $\X\in\R^{(T+1)n_u\times (T+1)n_x}$ be a $(n_u,n_x)$-block-lower triangular matrix. A pair of matrices $(\D,\E)\in\R^{(T+1)n_u\times r}\times \R^{r \times (T+1)n_x}$ is an \emph{$\epsilon$-causal factorization of $\X$ with band $r$} if they satisfy $\|\X - \D\E\|\leq\epsilon$ and Constraint~\ref{constraint:causality} (Causality).
\end{definition}
Note that a $0$-causal factorization is a causal factorization as defined in Definition~\ref{def:causal_facto}.

As stated in Theorem~\ref{thm:causalFactorization} below, Algorithm~\ref{algo:causalFactorization} computes an $\epsilon$-causal factorization.\footnote{The following convention is used: for two ``empty'' matrices $\mathbf{A}\in\R^{m\times 0}$ and $\mathbf{B}\in\R^{0\times n}$, $\mathbf{AB}\coloneqq0_{m\times n}$.} The algorithm works as follows: the encoder matrix $\E$ is constructed row by row according to the following rule: a row of $\X$ is added to $\E$ if and only if it can not be approximated by a linear combination of the rows already in $\E$. When a row is not added, the coefficients of the linear combination are added to the decoder $\D$.

\begin{algorithm}
\small
\caption{Approximate causal factorization} \label{algo:causalFactorization}
\begin{algorithmic}[1]
\Require $\X\in\R^{m\times n}$ block-lower triangular and $\epsilon\geq0$
\State $\D\coloneqq 0_{0\times0}$, $\E\coloneqq 0_{0\times n}$, $r\coloneqq0$ \Comment{$r$ is the band}
\For{$l=1,\dots,m$}
    \If{$\underset{d\in\R^{1\times r}}{\min}\left\|\X_{1:l,:} - \begin{bmatrix}\D\\ d\end{bmatrix}\E\right\|\leq\epsilon$} \label{algo:line:if}
        \State $\D\coloneqq \begin{bmatrix} \D \\ d\end{bmatrix}$ \label{algo:line:in-if}
    \Else \label{algo:line:else}
        \State $\E\coloneqq \begin{bmatrix} \E \\ \X_{l,:}\end{bmatrix}$, $\D\coloneqq \begin{bmatrix} \D & 0 \\ 0 & 1 \end{bmatrix}$, $r\coloneqq r+1$ \label{algo:line:in-else}
    \EndIf \Comment{$\D\in\R^{l\times r}$ and $\E\in\R^{r\times n}$}
\EndFor
\State \Return $(\D,\E)$
\end{algorithmic}
\end{algorithm}

\begin{theorem}\label{thm:causalFactorization}
Algorithm~\ref{algo:causalFactorization} returns an $\epsilon$-causal factorization. If $\epsilon=0$, the band of the returned causal factorization is $\rank\X$. If $\epsilon\geq\|\X\|$, the band of the returned causal factorization is zero.
\end{theorem}

\begin{proof} 
Let $(\D,\E)$ be the pair returned by Algorithm~\ref{algo:causalFactorization} and let $r$ be its band. First we prove that $(\D,\E)$ satisfies Constraint~\ref{constraint:causality}. Let $1\leq l_1<l_2<\dots<l_r\leq m$ be the values of $l$ such that the \textbf{else} statement on line~\ref{algo:line:else} is executed. For $k\in\{1,\dots,r\}$, let $t_k$ be such that $l_k=t_k n_u+i_k$ for some $i_k\in\{1,\dots,n_u\}$.

Then, for any $\tau>t_k$ and any $j\in\{1,\dots, n_x\}$, thanks to line~\ref{algo:line:in-else} and the block-lower triangularity of $\X$, we have $\E_{k,\tau n_x+j}=\X_{l_k,\tau n_x+j}=\X_{t_k n_u+i_k,\tau n_x+j}=0$. This proves the causality constraint on $\E$. Similarly, from line~\ref{algo:line:in-else}, it follows that $\D_{1:l_k,k}=\begin{bmatrix} 0_{(l_k-1)\times1} \\ 1 \end{bmatrix}$. Then, for any $t<t_k$ and $i\in\{1,\dots,n_u\}$, we have $t n_u+i<t_k n_u+i_k=l_k$ and $\D_{t n_u+i,k}=0$. This proves the causality constraint on $\D$.

Let $\D^{(l)}$ and $\E^{(l)}$ be the values of $\D$ and $\E$ at the end of the $l$-th execution of the \textbf{for} loop. Note that, at the end of algorithm, we have $\D = \D^{(m)}$ and $\E = \E^{(m)}$. To prove that $\|\X-\D\E\| \leq\epsilon$, we show that
\begin{align}\label{eq:proof:recursive}
\|\X_{1:l}-\D^{(l)}\E^{(l)}\|\leq\epsilon,
\end{align}
for all $l\in\{1,\dots,m\}$. We proceed inductively on $l$. For $l=1$, the \textbf{if} condition on line~\ref{algo:line:if} can be rewritten $\|\X_{1,:}\|\leq \epsilon$. If this condition holds, then $\D^{(1)}=0_{1\times0}$ and $\E^{(1)}=0_{0\times n}$ and \eqref{eq:proof:recursive} holds. If the \textbf{if} condition does not hold, then $\E^{(1)}=\X_{1,:}$ and $\D^{(1)}=1$ and \eqref{eq:proof:recursive} holds. For $l>1$, if the \textbf{if} condition holds, it directly implies \eqref{eq:proof:recursive}; if the \textbf{if} condition does not hold, then using the inductive assumption,
\begin{small}
\begin{align*}
\|\X_{1:l,:}-\D^{(l)}\E^{(l)}\|&=\left\| \begin{bmatrix} \X_{1:l-1,:} \\ \X_{l,:}\end{bmatrix} - \begin{bmatrix} \D^{(l-1)} & 0 \\ 0 & 1 \end{bmatrix} \begin{bmatrix} \E^{(l-1)} \\ \X_{l,:}\end{bmatrix} \right\|\\
&=\left\|\begin{bmatrix} \X_{1:l-1}-\D^{(l-1)}\E^{(l-1)} \\ 0 \end{bmatrix}\right\| \\
&=\left\| \X_{1:l-1}-\D^{(l-1)}\E^{(l-1)} \right\| \leq\epsilon.
\end{align*}
\end{small}
This concludes the proof of \eqref{eq:proof:recursive} and the proof of $\|\X-\D\E\|\leq\epsilon$.

Let us prove that when $\epsilon=0$, $r=\rank \X$. In that case the \textbf{if} condition in line~\ref{algo:line:if} is equivalent to the existence of a vector $d$ such that $\begin{bmatrix}\X_{1:l-1,:} \\ \X_{l,:}\end{bmatrix}=\begin{bmatrix}\D^{(l-1)} \\ d\end{bmatrix}\E^{(l-1)}$. Since $\X_{1:l-1,:}=\D^{(l-1)}\E^{(l-1)}$ from \eqref{eq:proof:recursive}, the \textbf{if} condition is equivalent to $\X_{l,:}\in\operatorname{range}~\E^{(l-1)}$. Consequently, a row $\X_{l,:}$ is added to $\E$ if and only if it is independent of the ones already added. It follows that $\E$ has $\rank \X$ rows and $r=\rank \X$.

Finally, if $\|\X\|\leq\epsilon$, since for all $l$, $\|\X_{1:l,:}\|\leq\|\X\|\leq\epsilon$, the \textbf{if} condition on line~\ref{algo:line:if} always holds with $\E\coloneqq 0_{0\times n}$. In this case, it follows that the final value of $r$ is zero.
\end{proof}

Note that the approximate causal factorization returned by Algorithm~\ref{algo:causalFactorization} is not necessarily optimal when $0<\epsilon<\|\mathbf{X}\|$, in the sense that there may exist another $\epsilon$-causal factorization with a smaller band than the one returned. However, minimizing the band is computationally hard since it involves a combinatorial problem in selecting the transmission times $t_k$. We also note that the optimal band is lower bounded by $i^*$, where $i^*$ is the index of the smallest singular value of $\X$ that is strictly greater than $\epsilon$, which is essentially the minimum rank SVD approximation that ignores causality. This fact can be used to get certificates of optimality for Algorithm~\ref{algo:causalFactorization}.

\subsection{Robustness to factorization error}\label{sec:robustness}

In this subsection, we derive an upper bound on the degradation of the L2 gain with respect to the factorization error. To this end, we first prove that an exact causal factorization of a $\K$ can be obtained from an approximate causal factorization of $\Phi_u$. Then, we bound the degradation of the L2 gain as a function of the factorization error on $\Phi_u$.

\begin{lemma}\label{thm:causalFactorization:SLS}
Let the pair $(\D_{\Phi_u},\E_{\Phi_u})$ be an $\epsilon$-causal factorization of a matrix $\Phi_u$, and let the matrix $\Phi_x$ be a $(n_x,n_x)$-block-lower triangular and invertible matrix. Then the pair $(\D_{\Phi_u},\E_{\Phi_u}\Phi_x^{-1})$ is an exact causal factorization of the matrix $\K\coloneqq\D_{\Phi_u}\E_{\Phi_u}\Phi_x^{-1}$ with same band.
\end{lemma}

\begin{proof}
First, note that the pair $(\D_\K,\E_\K) \coloneqq (\D_{\Phi_u}, \E_{\Phi_u}\Phi_x^{-1})$ is a factorization of $\K$, i.e., $\K=\D_\K\E_\K$. Then, we need to show that this pair satisfies the causality constraints. Since $(\D_{\Phi_u},\E_{\Phi_u})$ is an $\epsilon$-causal factorization, it satisfies the causality constraint for some $\{t_k\}_{k=1}^r$. Since $\D_\K=\D_{\Phi_u}$, $\D_\K$ satisfies the causality constraint as well. Let $k\in\{1,\dots,r\}$, $\tau>t_k$ and $j\in\{1,\dots,n_x\}$. To show that $\E_\K$ satisfies the causality constraint, we need to show $[\E_\K]_{k,\tau n_x+j}=0$. One can write $[\E_\K]_{k,\tau n_x+j}=\sum_{t=0}^T\sum_{i=1}^{n_x}[\E_{\Phi_u}]_{k,tn_x+i} [\Phi_x^{-1}]_{tn_x+i,\tau n_x+j}$. If $t>t_k$, then $[\E_{\Phi_u}]_{k,tn_x+i}=0$ for all $i$ thanks to causality. If $t\leq t_k$, then $t<\tau$ and $[\Phi_x^{-1}]_{tn_x+i,\tau n_x+j}=0$ for all $i$ because $\Phi_x^{-1}$ is $(n_x,n_x)$-block-lower triangular (since $\Phi_x$ is). It follows that all terms in the sum are zero and the proof is complete.
\end{proof}

The previous result allows to factorize $\Phi_u$ instead of $\K$. By doing so, we can use the robustness of SLS \cite[Theorem 2.2]{anderson2019system} to characterize the degradation of the L2 gain in terms of the factorization error. Note that \cite[Theorem 2.2]{anderson2019system} is usually used to study the degradation of performance with respect to an uncertainty on the system dynamics (namely, matrices $A$ and $B$). In what follows, we show that this theorem can be extended to bound the degradation of the L2 gain \emph{with respect to a perturbation of a feasible solution $\Phi_u$.} This is our second main result.

\begin{theorem}\label{thm:bound}
Let the matrices $\{\Phi_x,\Phi_u\}$ be $(n_x,n_x)$- and $(n_u,n_x)$-block-lower triangular (resp.) satisfying \eqref{eq:SLC} and
\begin{align}\label{eq:l2:robust}
\left\| \begin{bmatrix} \mathcal{Q}^{1/2} & 0 \\ 0 & \mathcal{R}^{1/2} \end{bmatrix} \begin{bmatrix}\Phi_x \\ \Phi_u\end{bmatrix}\right\|\leq \gamma_\epsilon\coloneqq \gamma/\beta_\epsilon-\alpha_\epsilon,
\end{align}
with $\alpha_\epsilon\coloneqq\|\mathcal{R}^{1/2}\|\epsilon$, and $\beta_\epsilon\coloneqq\|\mathcal{D}\|\sum_{t=0}^T\left(\|Z\mathcal{B}\|\epsilon\right)^t$ for some non-negative $\epsilon$ and $\gamma$. Let the pair $(\D_{\Phi_u},\E_{\Phi_u})$ be an $\epsilon$-causal factorization of $\Phi_u$, then the pair $(\D_{\Phi_u},\E_{\Phi_u}\Phi_x^{-1})$ is an (exact) causal factorization of a matrix $\K\coloneqq \D_{\Phi_u}\E_{\Phi_u}\Phi_x^{-1} $ that reaches an L2 gain not greater than $\gamma$.
\end{theorem}
\begin{proof}
Let $\tilde{\Phi}_u\coloneqq \D_{\Phi_u}\E_{\Phi_u}$. By definition of $\epsilon$-causal factorization, one can write $\tilde{\Phi}_u=\Phi_u+\Delta$ with $\Delta$ being $(n_u,n_x)$-block-lower triangular and such that $\|\Delta\|\leq\epsilon$. Then, since $\{\Phi_x,\Phi_u\}$ satisfy \eqref{eq:SLC}, we can write $\begin{bmatrix}I-Z\mathcal{A} & -Z\mathcal{B}\end{bmatrix}\begin{bmatrix}\Phi_x \\ \tilde{\Phi}_u\end{bmatrix} = I - Z\mathcal{B}\Delta$. Then, it follows from \cite[Theorem 2.2]{anderson2019system} (robustness of state feedback SLS), that the closed loop response of the system using the controller $\tilde{\K}=\tilde{\Phi}_u\Phi_x^{-1}$ leads to the closed loop response
$$
\begin{bmatrix} \mathbf{x} \\ \mathbf{u} \end{bmatrix} = \begin{bmatrix}\Phi_x \\ \tilde{\Phi}_u\end{bmatrix}(I-Z\mathcal{B}\Delta)^{-1}\mathcal{D}\mathbf{w}.
$$
Consequently, it induces the L2 gain
\begin{align*}
&\left\| \begin{bmatrix} \mathcal{Q}^{1/2} & 0 \\ 0 & \mathcal{R}^{1/2} \end{bmatrix} \begin{bmatrix}\Phi_x \\ \Phi_u+\Delta\end{bmatrix}(I-Z\mathcal{B}\Delta)^{-1}\mathcal{D} \right\|.
\end{align*}
Since $Z\mathcal{B}\Delta$ is block-lower triangular, we have $(I-Z\mathcal{B}\Delta)^{-1}=\sum_{t=0}^{T} (Z\mathcal{B}\Delta)^t$. Using triangular inequality and submultiplicativity of the induced $2$-norm, the L2 gain can be bounded as
\begin{align*}
& \left\| \left( \begin{bmatrix} \mathcal{Q}^{1/2} & 0 \\ 0 & \mathcal{R}^{1/2} \end{bmatrix} \begin{bmatrix}\Phi_x \\ \Phi_u\end{bmatrix} + \begin{bmatrix}0\\ \mathcal{R}^{1/2}\Delta \end{bmatrix} \right)\sum_{t=0}^{T} (Z\mathcal{B}\Delta)^t\mathcal{D} \right\| \\
~
\leq& \left(\left\| \begin{bmatrix} \mathcal{Q}^{1/2} & 0 \\ 0 & \mathcal{R}^{1/2} \end{bmatrix} \begin{bmatrix}\Phi_x \\ \Phi_u\end{bmatrix}\right\| +\left\|\begin{bmatrix}0\\ \mathcal{R}^{1/2}\Delta \end{bmatrix} \right\|\right)\\
~
&\hspace{5cm}\times \left\| \sum_{t=0}^{T}  (Z\mathcal{B}\Delta)^t\mathcal{D}\right\| \\
~
\leq& ( \gamma_\epsilon + \|\mathcal{R}^{1/2}\|\|\Delta\|) \sum_{t=0}^{T} \left(\| Z\mathcal{B}\|\|\Delta\|\right)^t\left\|\mathcal{D}\right\| \leq \gamma,
\end{align*}
where the last inequality uses $\|\Delta\|\leq\epsilon$.

Finally, it follows from Lemma~\ref{thm:causalFactorization:SLS} that the pair $(\D_{\Phi_u},\E_{\Phi_u}\Phi_x^{-1})$ is an (exact) causal factorization of $\K$.
\end{proof}

Overall, Theorem~\ref{thm:bound} allows to approximately solve Problem~\ref{prob} via the following procedure:
\begin{enumerate}
\item Fix an $\epsilon\geq0$.
\item Find block-lower triangular matrices $\{\Phi_x,\Phi_u\}$ that approximately solve 
\begin{align} \label{eq:prob:epsilon}
\min_{\Phi_x,\Phi_u} \rank \Phi_u \text{ s.t. \eqref{eq:SLC}, \eqref{eq:l2:robust}},
\end{align}
for example via nuclear norm minimization.
\item Compute an $\epsilon$-causal factorization $(\D_{\Phi_u},\E_{\Phi_u})$ of $\Phi_u$ using Algorithm~\ref{algo:causalFactorization}.
\item Implement the controller with the pair $(\D_\K,\E_\K)\coloneqq(\D_{\Phi_u},\E_{\Phi_u}\Phi_x^{-1})$ which reaches an L2 gain not greater than $\gamma$.
\end{enumerate}

If $\epsilon$ is chosen too large, the corresponding $\gamma_\epsilon$ is too small and can lead to infeasibility when solving problem~\eqref{eq:prob:epsilon}. On the opposite, if $\epsilon$ is too small, then the $\epsilon$-causal factorization can give a band higher than the numerical rank, leading to extra transmissions. However, in the next section, we show that for a large range of $\epsilon$, the procedure above achieves good performance.

\section{NUMERICAL EXPERIMENTS}

To illustrate our method,\footnote{A \textsc{Python} code that implements our method and generates the figures is available at: \url{https://github.com/aaspeel/Min-L2-Consistent}. The reported computation times are obtained using a laptop with an Apple M2 chip and 16 GB of RAM.} we consider a two-dimensional double integrator dynamics $\ddot{p}_x=u_x$ and $\ddot{p}_y=u_y$ discretized using forward Euler method with time step $0.1$. This gives the state $x_t=\begin{bmatrix}p_x & \dot{p}_x & p_y & \dot{p}_y\end{bmatrix}^\top$. A noise $w_t$ is then added to the dynamics with $D=I$. Matrices $Q$ and $R$ are chosen to be the identity as well, and $T=20$ time steps are considered. Except stated otherwise, $\gamma=9.344$ (which is close to the infimum L2 gain that can be reached with periodic measurements of period 3), and our method uses $\epsilon=10^{-8}$. To (approximately) minimize the rank, we do 8 iterations of the reweighted nuclear norm heuristic described in~\cite[Section~III]{mohan2010reweighted} with regularization parameter $\delta=0.01$. We compare our method with the one presented in \cite{balaghi20192}, referred to as the \emph{minimax} method here. Note that the minimax method does not use an encoder and communicates the state $x_t\in\R^4$, which requires $n_x=4$ messages.

Fig.~\ref{fig:communication_times} compares the transmission times obtained with our method with the ones obtained using the minimax method and the ones of period 3 (this is the largest period that leads to an L2 gain smaller than $\gamma$). Our method gives a controller that requires 8 transmissions, while the minimax method requires 16. However, note that our method is significantly slower: it synthesizes a controller in 376 seconds (375 seconds for the nuclear norm minimization and 1 second for the approximate causal factorization) while the minimax method runs in less than 1 second. Note that these computations are performed offline.

\begin{figure}[!ht]
    \centering
    \includegraphics[width=.9\columnwidth]{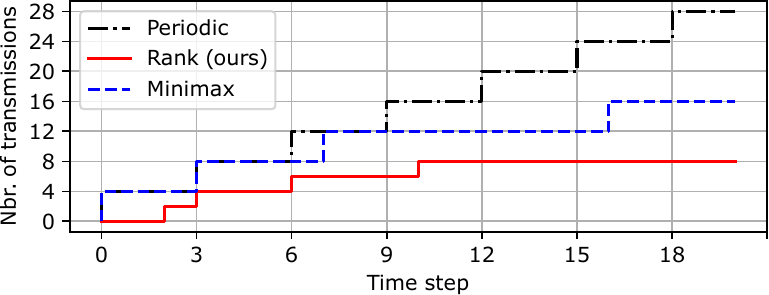}
    \vspace{-0cm}
    \caption{Transmission times obtained with our method and with the minimax method. Periodic state transmission is also scheduled for the period 3.}
    \vspace{-.1cm}
    \label{fig:communication_times}
\end{figure}

Fig.~\ref{fig:causal_factorization} shows the sparsity pattern of the matrix $\K$ and its (exact) causal factorization $(\D,\E)$. The staircase structure of matrices $\D$ and $\E$ comes from the causality Constraint~\ref{constraint:causality}. The number of transmissions obtained for different values of the L2 bound $\gamma$ is illustrated in Fig.~\ref{fig:communications_vs_gamma}. Our method consistently leads to fewer transmissions than the minimax method.

\begin{figure}[!ht]
    \centering
    \includegraphics[width=\columnwidth]{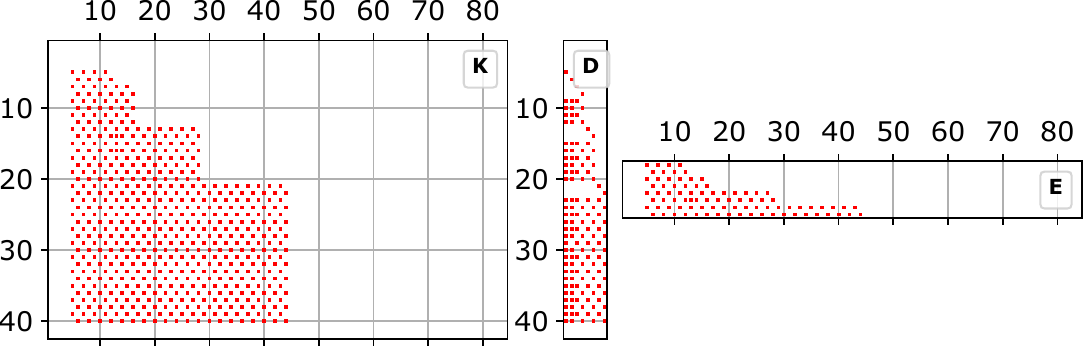}
    \vspace{-0.3cm}
    \caption{Sparsity pattern of $\K$ and its causal factorization $(\D,\E)$.}
    \vspace{-.1cm}
    \label{fig:causal_factorization}
\end{figure}

\begin{figure}[!ht]
    \centering
    \includegraphics[width=.9\columnwidth]{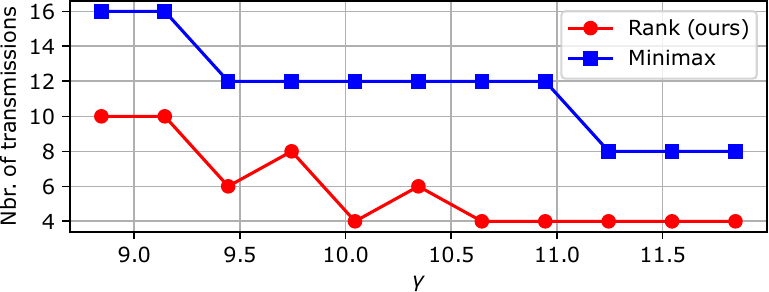}
    \vspace{-0cm}
    \caption{Number of transmissions with respect to the bound $\gamma$ on the L2 gain.}
    \vspace{-.1cm}
    \label{fig:communications_vs_gamma}
\end{figure}

Finally, we analyze the sensitivity of our method with respect to $\epsilon$. Our method leads to 8 transmissions for all $\epsilon\in\{10^{-9}, 10^{-7}, 10^{-5}, 10^{-3} \}$, and to 22 transmissions for $\epsilon=10^{-12}$. This shows that our method gives good performance for a large range of $\epsilon$. This is of practical interest since it means that $\epsilon$ does not require to be fine tuned.

\section{CONCLUSIONS AND FUTURE WORKS}

In this work, we addressed the problem of designing a controller that requires a minimum number of sensor-to-actuator transmissions while satisfying a bound on the L2 gain. To this end, we rely on causal factorization and system level synthesis to reduce the problem to a rank minimization over a convex set. To anticipate the fact that optimization methods yield matrices that are only numerically low rank, we introduce the notion of approximate causal factorization and an algorithm to compute it. Finally, the degradation in L2 gain caused by the factorization error is bounded.

Our future works will consider an event-triggered extension suitable to the infinite horizon setting.

\bibliographystyle{IEEEtran}
\bibliography{bibliography_short}

\end{document}